\documentclass[12pt, a4paper]{extarticle}
\usepackage[margin=2cm]{geometry}
\usepackage[affil-it]{authblk}

\usepackage[english]{babel}
\usepackage{graphicx} 
\usepackage{wrapfig}
\usepackage{amsmath,amsthm,amssymb,scrextend}
\usepackage{booktabs}            
\usepackage{multirow}            
\usepackage{amsfonts}            
\usepackage{multicol}
\usepackage[shortlabels]{enumitem}
\usepackage{cutwin}
\usepackage{float}
\usepackage{csquotes}
\usepackage{enumitem}
\usepackage{tikz}
\usepackage{hyperref}

\def \bf#1 {\textbf{#1 }}
\def \sumt {\sum\limits}
\providecommand{\keywords}[1]
{
  \small	
  \textbf{\textit{Keywords:}} #1
}

\renewenvironment{proof}{\begin{addmargin}[1em]{0em}\begin{newproof}}{\end{newproof}\end{addmargin}\qed}
\newtheorem{thm}{Theorem}

\newtheorem*{lm*}{Lemma}
\newtheorem{defin}{Definition}
\newtheorem*{defin*}{Definition}

\def \sumt {\sum\limits}

\DeclareMathOperator{\out}{out}
\DeclareMathOperator{\xor}{xor}
\DeclareMathOperator{\prsum}{sum}
\def \cN {\mathcal{N}}

\def \l {\lambda}

\def \d {\partial}

\begin{document}

\title{A series of real networks invariants}
\author{Mikhail Tuzhilin%
  \thanks{Affiliation: HSE University, Moscow State University, Electronic address: \texttt{mikhail.tuzhilin@math.msu.ru};}
}
\date{}
\maketitle

\begin{abstract}
In this article, we propose a generalization of two well-known invariants of real networks: degree and ksi-centrality. More precisely, we found a set of centralities (degree centrality and ksi-centrality for the cases $j = 0, 1$) based on the Laplacian matrix that closely match the Weibull distribution; for real networks, they are right-skewed, while for artificial networks, they are centered. A threshold of 1 in the Pearson skewness coefficients accurately distinguishes these real networks from artificial ones.
\end{abstract}

\keywords{Centralities, local and global characteristics of networks, Laplacian}


\section{Introduction}

One of the most well-known invariant of real networks or networks based on a real-world data is the degree centrality. Albert, Jeong, and Barabasi~\cite{BA1} found scale-free property of real networks: for real networks the degree centrality has power-low distribution and for random networks this distribution is different. Also they provided an algorithm~\cite{BA2} to construct a network that satisfies scale free property, however the average clustering coefficient of this network was insufficient, i.e. the network was not small-world. Boccaletti, Hwang, and Latora~\cite{BH} propose an algorithm to construct a scale-free network with a large average clustering coefficient. In~\cite{Tuz} author found that small-world and scale-free properties are insufficient for real-world networks. More precisely, the author proposed a new measure of centrality, called ksi-centrality, whose distribution well fitted by the Weibull distribution, has a linear structure with a "heavy tail" on the log plot, is right-skewed for real networks, and is central for  networks models (random, Barabási-Albert, Watts-Strogatz and Boccaletti-Hwang-Latora), which in turn shows that the ksi-centrality distinguishes real networks from artificial. 

In this paper, we constructed a series of invariants based on degree and ksi-centrality that are also well described by the Weibull distribution and have similar properties for distinguishing between real and artificial networks. These new invariants are based on Laplacian matrix and connected to a Laplacian matrix spectrum.

\section{Generalization of degree and ksi centralities}

\subsection{Theory}

Consider a graph $G$ as 1-dimensional simplicial complex. Let $G$ be simple undirected connected graph for simplicity. To define chain complex structure we should add orientation for $G$. Since our definition won't depend on orientation, we will choose orientation later. Let $C_k$ be the $k$-dimensional chain group. For this graph we have only two non-trivial dimensions: $C_0$ and $C_1$. The boundary map $C_0\xleftarrow{\d_1} C_1$ generates the map between adjoint spaces (i.e. linear functions on chains) $C^0\xrightarrow{\d^0} C^1$.

Let $(\cdot, \cdot)$ be a scalar product on $k$-dimensional chains $C_k$. By this product one can determine an adjoint map $C^0\xleftarrow{(\d^0)^*} C^1$ by flowing: for any two $f,g\in C^k,\; (\d^kf,g) = \big(f,(\d^k)^* g\big)$. We have
$$
C_0\xrightarrow{\d_1} C_1,
$$
$$
C^0\underset{(\d^0)^*}{\overset{\d^0}{\rightleftarrows}} C^1.
$$
The map $L = (\d^0)^*\d^0:C^0\rightarrow C^0$ is called Laplace-Beltrami operator and it corresponds to Laplacian matrix in the orthonormal basis in $C^0$.

Let's denote for some vertex $i$ the set of its $j$'s neighbors (i.e. all vertices precisely on the distance $j$ from $i$) by $\cN^j(i)$. Let's for fixed $i$ consider a characteristic function of its $j$'s neighbors $\chi^j_i$, that is 
$$
    \chi^j_i(k) = \begin{cases}
        1, & k\in \cN^j(i), \\
        0, & \text{otherwise.}
    \end{cases}
$$

We can consider this function $\chi^j_i$ as element of $C^0$, since it is defined us function on vertices. Note that any centrality can be considered as element of $C^0$. Also note that $\sum_j \chi^j_i\in \ker(L)$.

\begin{defin}
    Let's call \bf{$j$-neighborhood centrality} the following
    $$  
         \xi^j_{i} = \frac {\Big(L\chi^j_i, \chi^j_i\Big)} {\big(\chi^j_i, \chi^j_i\big)}, \;\text{for }i\in V(G).
    $$
\end{defin}

Let's consider some examples. 
\begin{enumerate}
    \item For the case $j=0$ the characteristic function $\chi^0_i$ is the characteristic function of the vertex $i$. Now let's fix $i$ and determine the orientation of $G$. Consider any orientation on induced subgraphs $\cN^j(i)$ and orient edges starting from $\cN^j(i)$ and ending in $\cN^{j+1}(i)$ for any $j$. If $j = 0$ the image $\d^0\big(\chi^0_i\big)$ is the element of $C^1$ which equals $-1$ on edges outer edges from the vertex $i$ and $0$ otherwise (see figure~\ref{fig1}). The image $(\d^0)^*\d^0\big(\chi^0_i\big)\in C^0$ and it is a function which equals $\deg(i)$ on the vertex $i$, $-1$ on its neighbors and $0$ otherwise. Hence, $\big(L\chi^0_i, \chi^0_i\big) = \deg(i)$ and $\xi^0_i = \deg(i)$. The Laplacian map $L:C^0\rightarrow C^0$ doesn't depend on orientation of a graph. $\big|E\big(\cdot\,,\, \cN^1(i)\cup\{i\}\big)\big|$

\begin{figure}[H]
    \centering
    \includegraphics[width=0.8\linewidth]{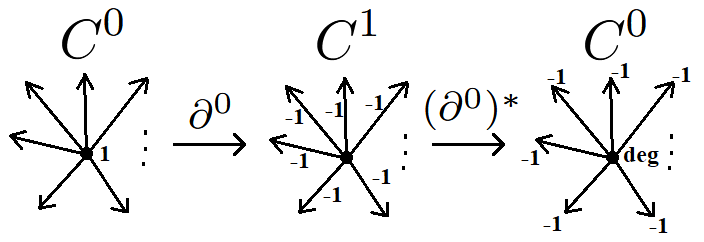}
    \caption{Image of Laplacian-Beltrami operator for the function $\chi^0_i$.}
    \label{fig1}
\end{figure}
    
    \item For the case $j=1$ the characteristic function $\chi^1_i$ is the characteristic function of adjacent vertices of $i$. The image $\d^0\big(\chi^1_i\big)\in C^1$ is the function which equals $1$ on edges outer edges from the vertex $i$, $-1$ on outer edges from $\cN^1(i)$ and $0$ otherwise (see figure~\ref{fig2}). The image $(\d^0)^*\d^0\big(\chi^1_i\big)\in C^0$ and it is a function which equals $-\deg(i)$ on the vertex $i$, the number of outer connections from the first neighborhood $\big|E\big(k, \cN^2(i)\cup\{i\}\big)\big|$ for each vertex $k\in \cN^1(i)$, $-\big|E\big(k, \cN^1(i)\big)\big|$ for $k\in \cN^2(i)$ and $0$ otherwise. Hence, $\big(L\chi^1_i, \chi^1_i\big) = |E(\cN^1(i), V(G)\setminus\cN^1(i))|$ and $\big(\chi^1_i, \chi^1_i\big)$ is the number of the first neighbors $\cN^1(i)$ that is $\deg(i)$. Therefore, $\xi^1_i = \xi_i$ defined in the article~\cite{Tuz}. 

\begin{figure}[H]
    \centering
    \includegraphics[width=0.8\linewidth]{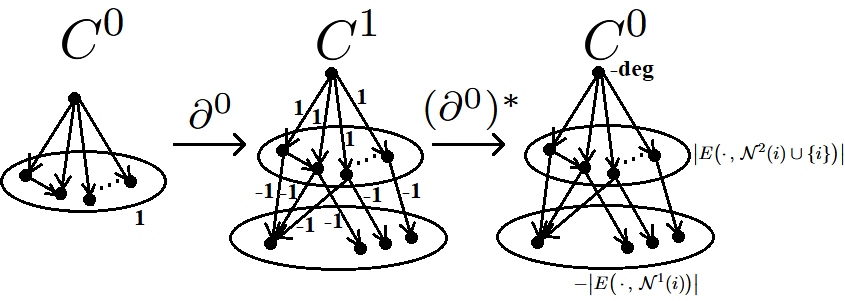}
    \caption{Image of Laplacian-Beltrami operator for the function $\chi^1_i$.}
    \label{fig2}
\end{figure}

\end{enumerate}

By these examples we see that this is a generalization of two invariants of real networks: degree and ksi-centrality. Let's note that this centrality can be also rewritten in terms of graph characteristics.

\begin{thm}\label{thm1}
    For any graph $j$-neighborhood centrality
    $$
        \xi^j_i = \frac {\big|\out\big(\cN^j(i)\big)\Big|} {\big|\cN^j(i)\big|},
    $$
    where $\out(H) = E(H, V(G)\setminus H)$ the set of outer connections from the set $H$ to other vertices.
\end{thm}
\begin{proof}
    It is very simple fact from the formula for Laplacian matrix bilinear form: $\big(L x, x\big) = \sumt_{k,l\in V(G), k\sim l}(x^k-x^l)^2$ and thus $(x_k-x_l)^2 = 1$ only if $x_k = 1$ (and thus $k\in \cN^j(i)$) and $x_l = 0$ (and thus $l\notin \cN^j(i), l\sim k$) and vice versa.
\end{proof}

For the case if graph has isolated vertex $i$ we define $\xi^j_i = 0$. Also let's note that these centralities are connected to the Chegeer number and  maximal eigenvalue of Laplacian matrix (by definitions).

\begin{thm}\label{st1}  Consider an undirected graph $G$. Let $h(G)$ be the Chegeer number of $G$. If $\big|\cN^j(i)\big|\leq\frac n 2$, then $\xi^j_i\geq h(G)$.
\end{thm}

Note that for real networks $\big|\cN^j(i)\big|\leq\frac n 2$ is the common case.

\begin{thm}\label{st1}  Consider an undirected graph $G$. Let $0=\l_1\leq\l_2\leq\cdots\leq\l_{n}$ be the Laplacian matrix spectrum. Then $\xi^j_i\leq\l_n$.
\end{thm}

\subsection{Calculations}

The theorem~\ref{thm1} can be used for calculation of $j$-neighborhood centralities sequentially by $j$. Let's store two characteristic matrices $\chi^{j-1}_i(k)$ and $\chi^{j-2}_i(k)$ in memory for $i, k\in V(G)$ from two pervious steps. Let $k$ be the column index. The algorithm is following:
\begin{enumerate}
    \item Initialization: $\chi^{0} = I, \chi^{1} = A$, where $I$ is the identity matrix and $A$ is adjacency matrix.
    \item Calculate the matrix $\chi^{j} = \xor\big(A \chi^{j-1}, \chi^{j-2}\big)$,
    \item For each $i$ calculate $\Big(L\chi^j_i, \chi^j_i\Big)$ and $\big(\chi^j_i, \chi^j_i\big) = \underset{k}{\prsum}(\chi^j_i)$.
\end{enumerate}

For $j = 0$ and $j = 1$ $j$-neighborhood centralities are degree and ksi centralities and they can be calculated separately. For calculations of ksi-centrality see~\cite{Tuz}. Note that calculations in this algorithm can be improved by storing additionally the matrix $A\chi^{j-1}$ which is calculated when the matrix $L\chi^j$ calculated on the previous step.


\section{Comparison of distributions for real networks and artificial}

We calculate distributions of $\xi^j$ centralities for $t=j=1,2,3,4$ based on the datasets from~\cite{Tuz} (real networks) and for Erdos-Renyi, Barabási-Albert, Watts-Strogatz and Boccaletti-Hwang-Latora (artificial). For real networks we observed similar behavior as for the degree distribution and the ksi distribution~\cite{Tuz}: these distributions are well approximated by the Weibull distribution, have a right-skewed skewness, and are nearly linear on the log-log plot with some noise and tail (Figures~\ref{f1}--\ref{f_last}). Furthermore, we observed the same result for Erdos-Rényi, Barabási-Albert, Watts-Strogatz, and Boccaletti-Hwang-Latora networks as in~\cite{Tuz}: the distributions are centered (figures~\ref{fER},~\ref{fWS}). We also observed the same result for the Pearson skewness coefficients: there is a gap between the real networks and the models for each invariant, and the same threshold level of 1 allows us to distinguish them (in the Boccaletti-Hwang-Latora model the fourth invariant turned out to be too noisy).

\section{Discussion}

In this article we constructed a sequence of centralities called $j$-neighborhood centralities based on Laplacian matrix or Laplacian-Beltrami operator for a graph. From the mathematical point of view these centralities have connections with  the Chegeer number and maximal eigenvalue of Laplacian matrix. From the point of view of applications these centralities are generalization of two invariants of real networks: degree centrality and ksi-centrality (cases $j=0, 1$). They have right-skewed distributions for real networks, different (centered) for artificial and well fitted with Weibull distribution. We have shown that a threshold value of 1 in Pearson's skewness coefficients distinguish real networks from artificial ones well for each invariant.

We call these $j$-neighborhood centralities invariants because for every network, regardless of its application domain (biological, social, internet, etc.), regardless of the ratio between the number of vertices and edges, the distributions of these centralities turn out to be right-skewed, well-fitted by the Weibull distribution, and also show a linear behavior with a heavy tail on the log plot (as is the case for the scale-free property). However, for the scale-free property, there is a strong theory explaining its foundations and a model for constructing networks with this property, but our centrality measures represent entirely new invariants, and many of them remain to be studied. Nevertheless, it is easy to apply them even now. For example, it can be applied for checking whether the current network is real-like or artificial or, since it is invariants of real networks, they are important candidates for features in GML models or graph-based machine learning models. 

The support from the Basic Research Program of HSE University is gratefully acknowledged.

\begin{figure}[h]
    \centering
    \includegraphics[width=\linewidth]{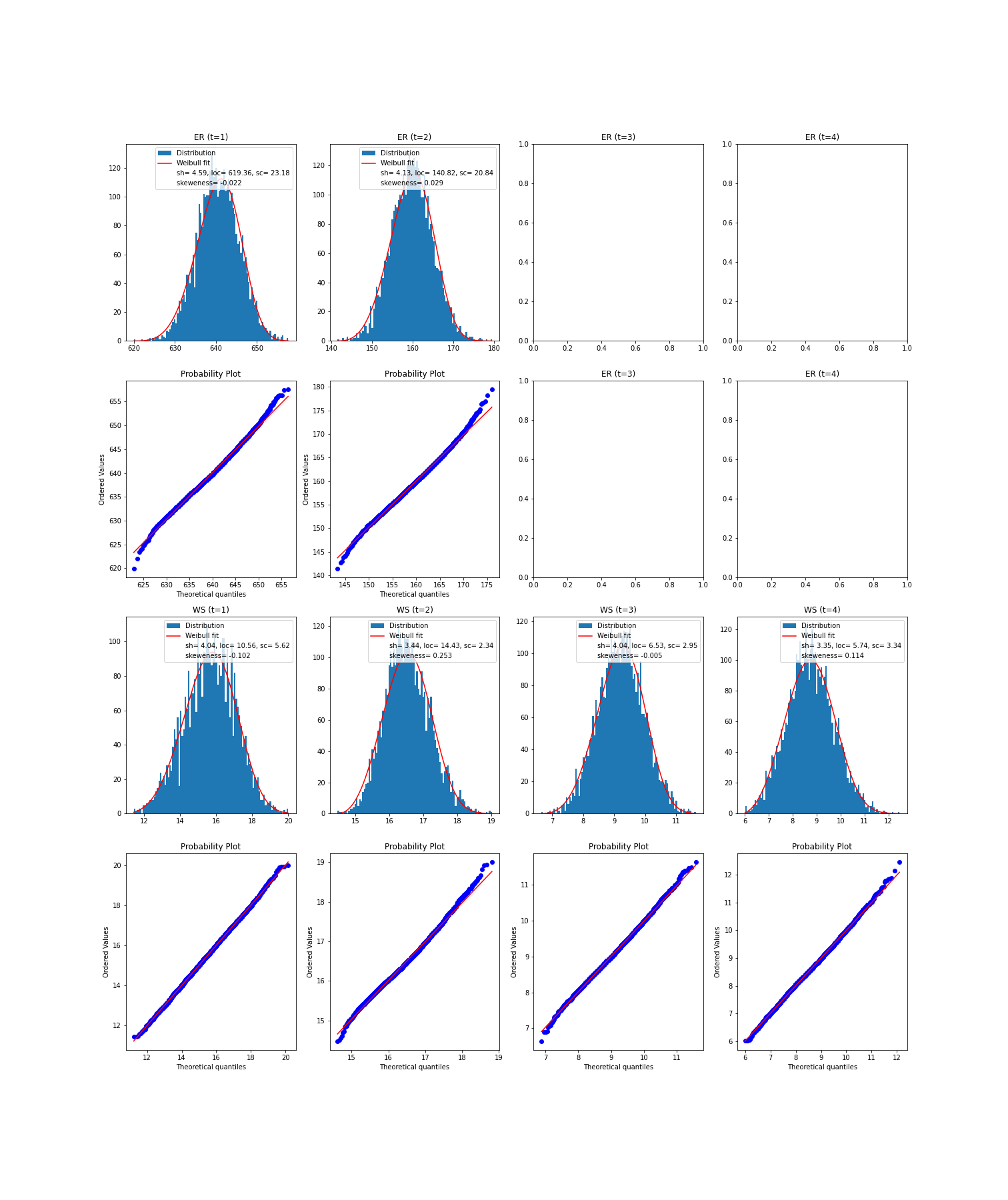}
    \caption{
    Distributions of $\xi^j$ centralities for $t=j=1,2,3,4$ for Erdos-Renyi $(n,p) = (4000, 0.02)$ and Watts-Strogatz $(n, k, p) = (4000, 21, 0.3)$ networks fitted with Weibull distribution and their Pearson's moment coefficients of skewness.}
    \label{fER}
\end{figure}

\begin{figure}[h]
    \centering
    \includegraphics[width=\linewidth]{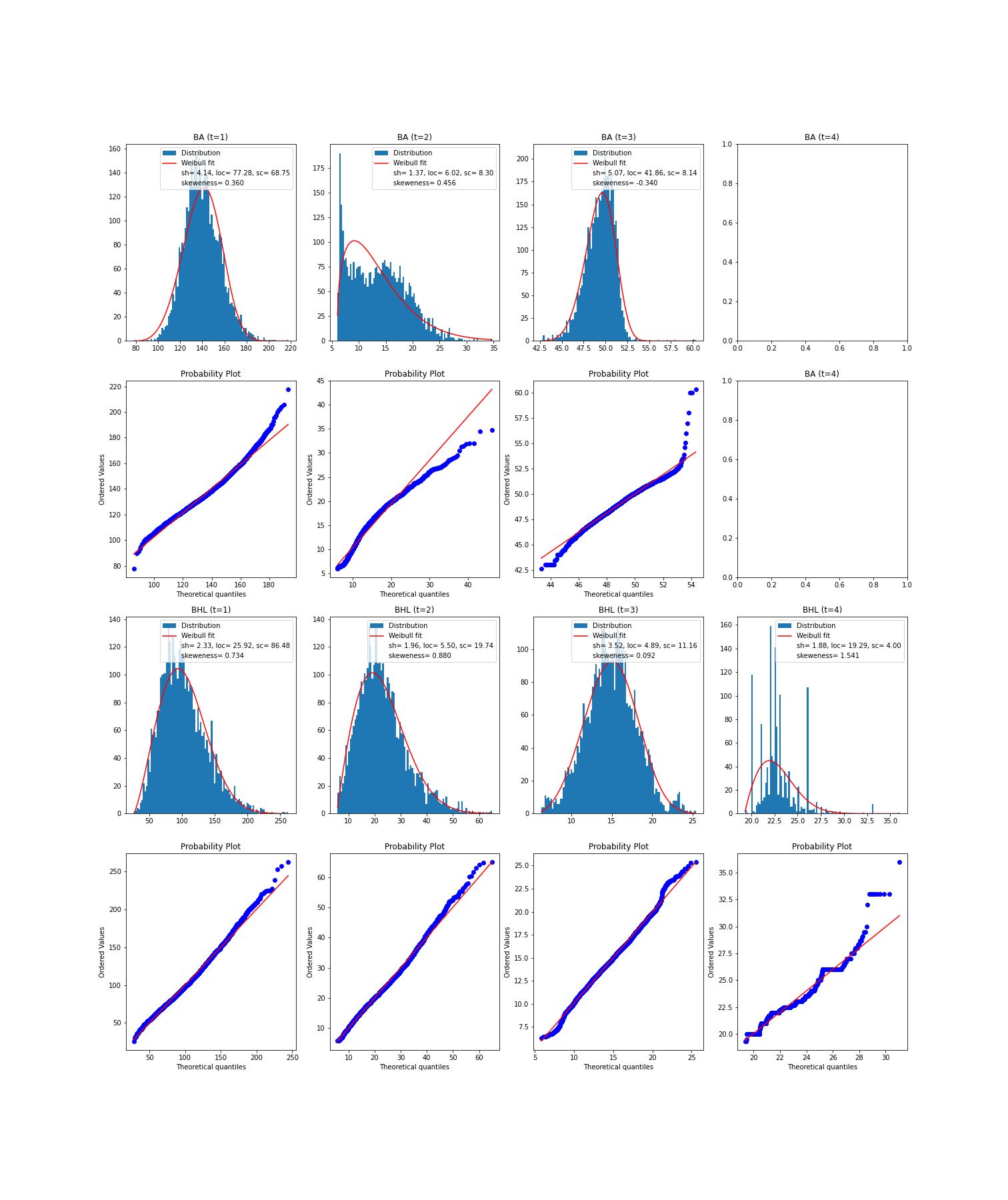}
    \caption{Distributions of $\xi^j$ centralities for $t=j=1,2,3,4$ for Barabasi-Albert $(n,m) = (4000, 43)$ and Boccaletti-Hwang-Latora $(n, m, n_0) = (4000, 20, 100)$ networks fitted with Weibull distribution and their Pearson's moment coefficients of skewness.}
    \label{fWS}
\end{figure}

\begin{figure}[h]
    \centering
    \includegraphics[width=\linewidth]{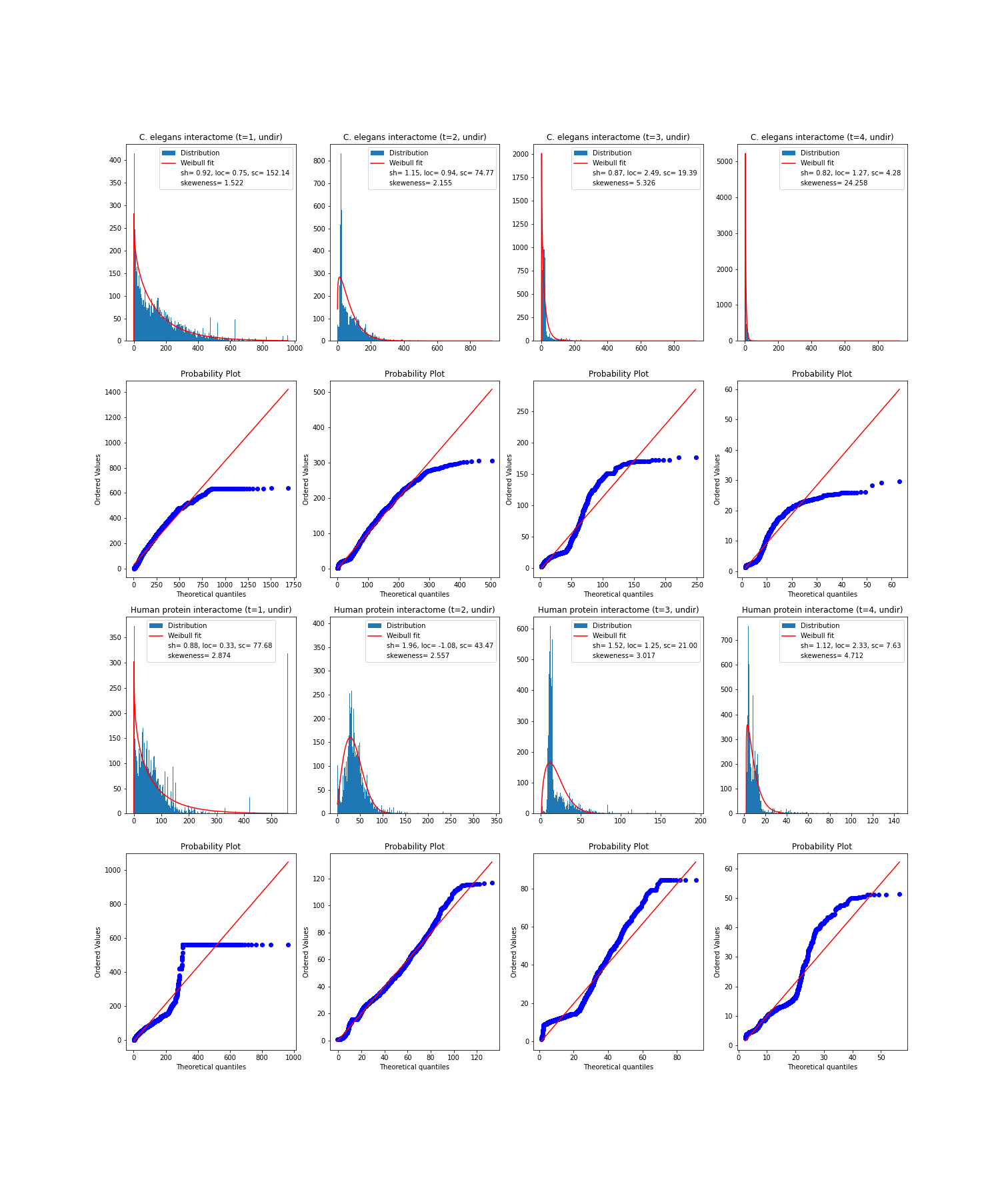}
    \caption{Distributions of $\xi^j$ centralities for $t=j=1,2,3,4$ for real networks data from~\cite{Tuz} fitted with Weibull distribution and their Pearson's moment coefficients of skewness.}
    \label{f1}
\end{figure}

\begin{figure}[h]
    \centering
    \includegraphics[width=\linewidth]{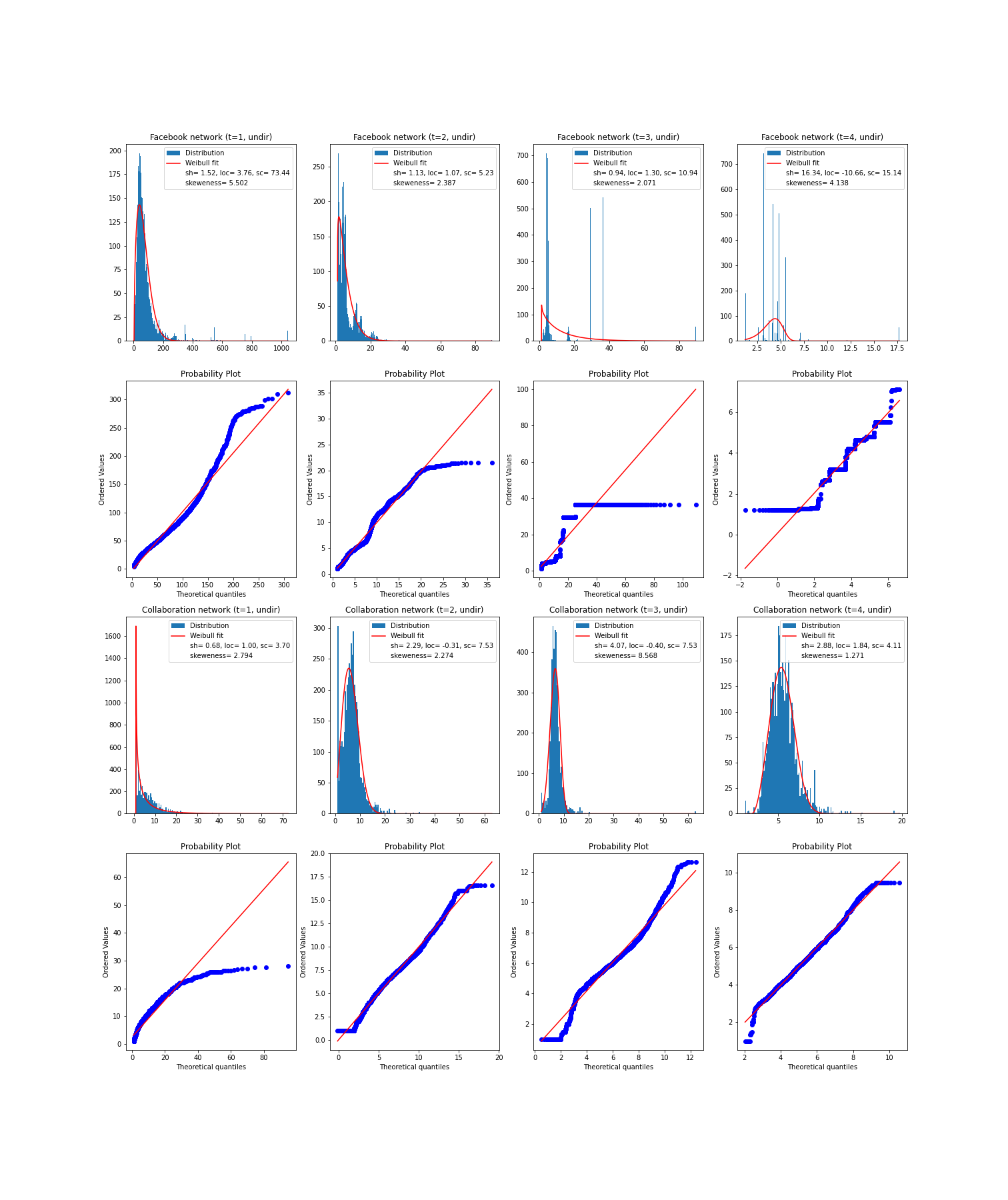}
    \caption{Distributions of $\xi^j$ centralities for $t=j=1,2,3,4$ for real networks data from~\cite{Tuz} fitted with Weibull distribution and their Pearson's moment coefficients of skewness.}
\end{figure}

\begin{figure}[h]
    \centering
    \includegraphics[width=\linewidth]{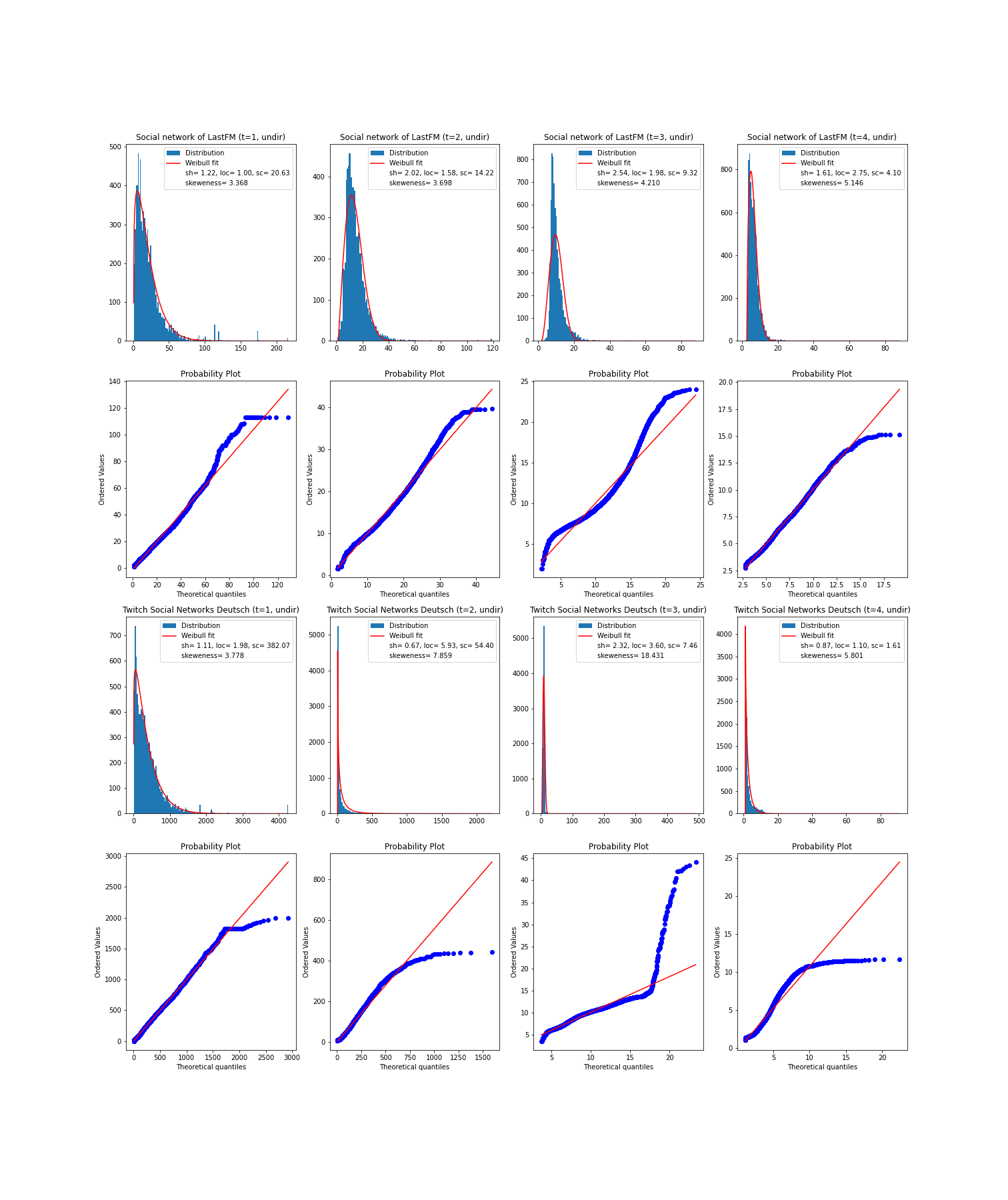}
    \caption{Distributions of $\xi^j$ centralities for $t=j=1,2,3,4$ for real networks data from~\cite{Tuz} fitted with Weibull distribution and their Pearson's moment coefficients of skewness.}
\end{figure}

\begin{figure}[h]
    \centering
    \includegraphics[width=\linewidth]{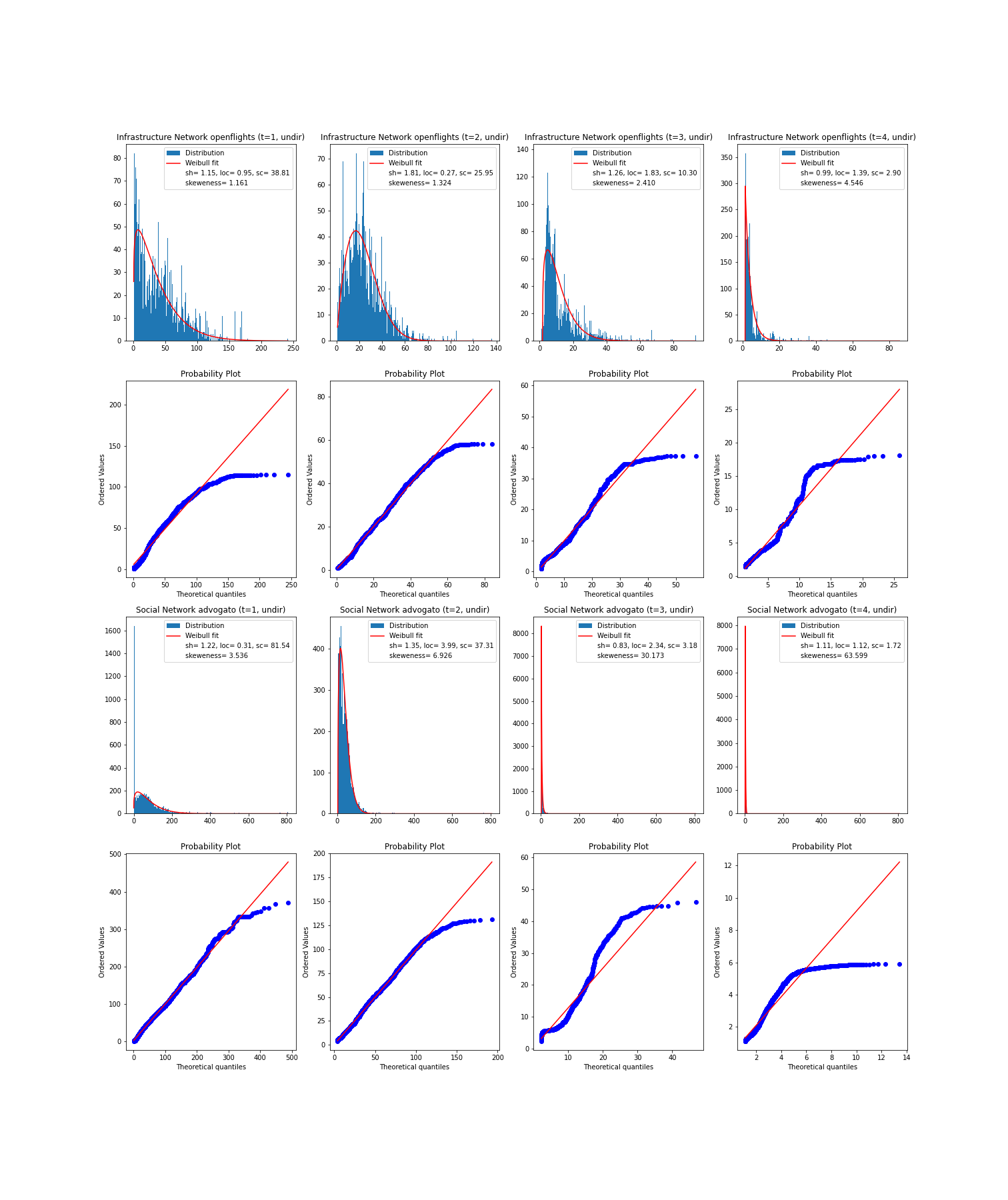}
    \caption{Distributions of $\xi^j$ centralities for $t=j=1,2,3,4$ for real networks data from~\cite{Tuz} fitted with Weibull distribution and their Pearson's moment coefficients of skewness.}
\end{figure}

\begin{figure}[h]
    \centering
    \includegraphics[width=\linewidth]{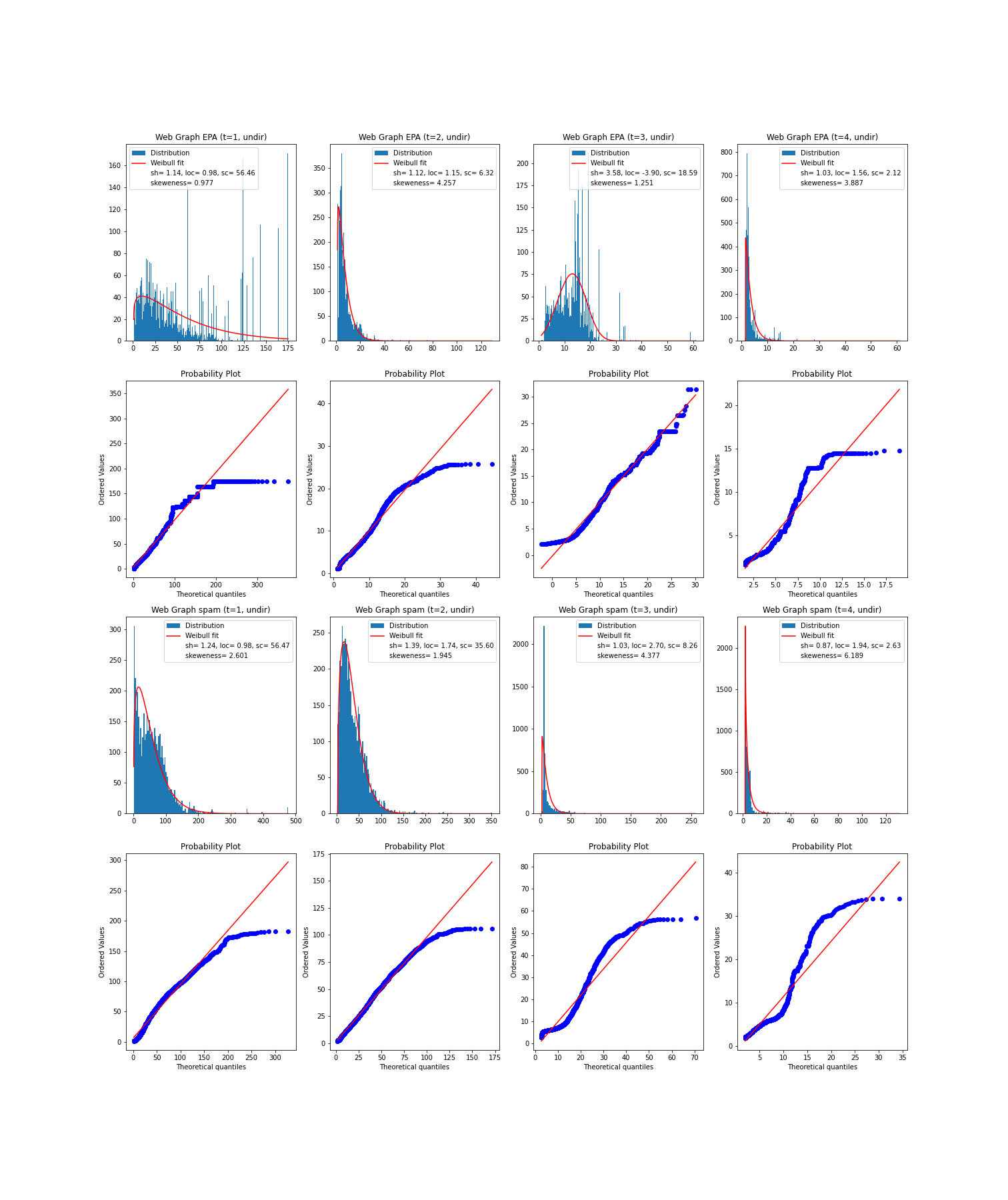}
    \caption{Distributions of $\xi^j$ centralities for $t=j=1,2,3,4$ for real networks data from~\cite{Tuz} fitted with Weibull distribution and their Pearson's moment coefficients of skewness..}
\end{figure}

\begin{figure}[h]
    \centering
    \includegraphics[width=\linewidth]{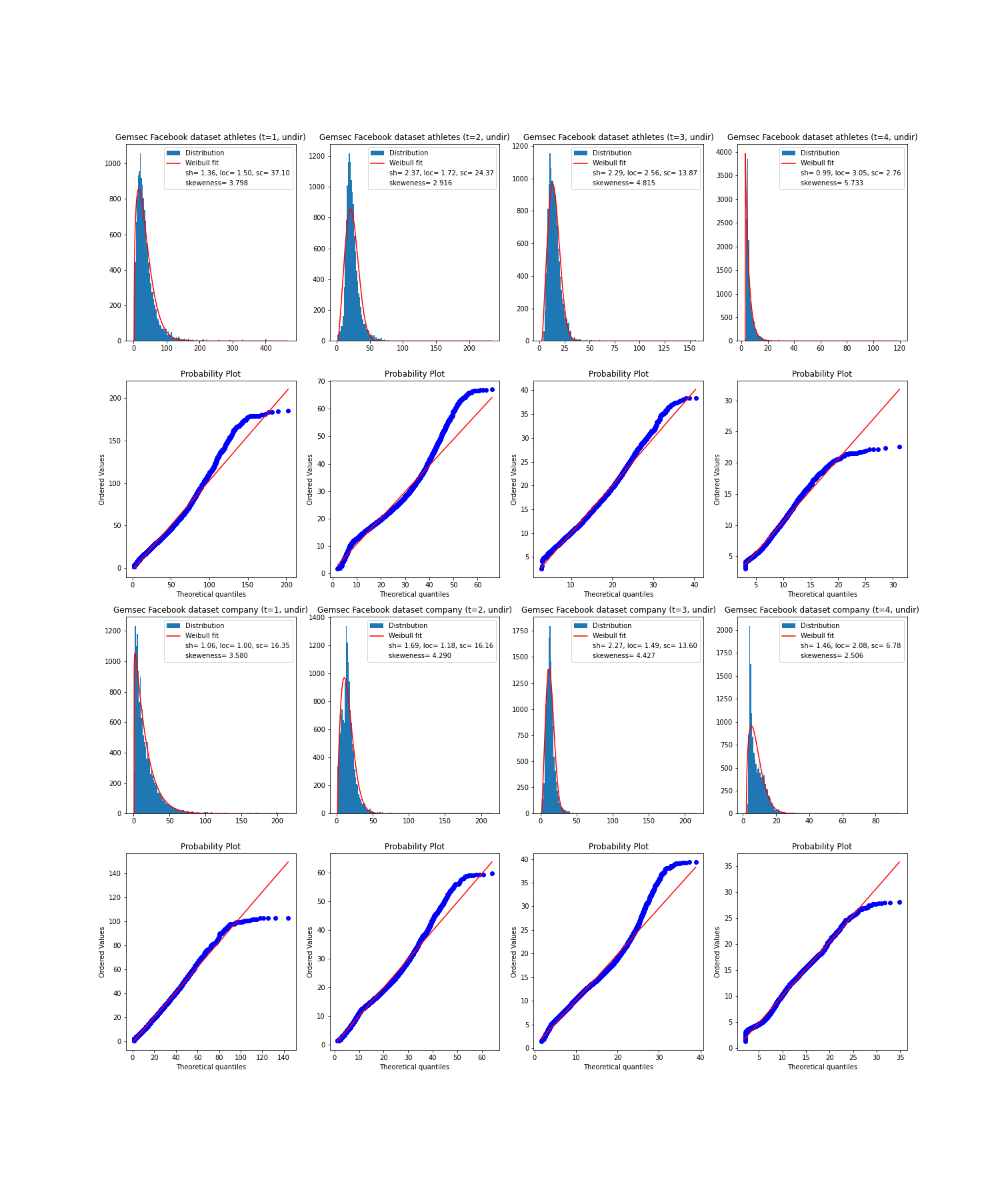}
    \caption{Distributions of $\xi^j$ centralities for $t=j=1,2,3,4$ for real networks data from~\cite{Tuz} fitted with Weibull distribution and their Pearson's moment coefficients of skewness.}
\end{figure}

\begin{figure}[h]
    \centering
    \includegraphics[width=\linewidth]{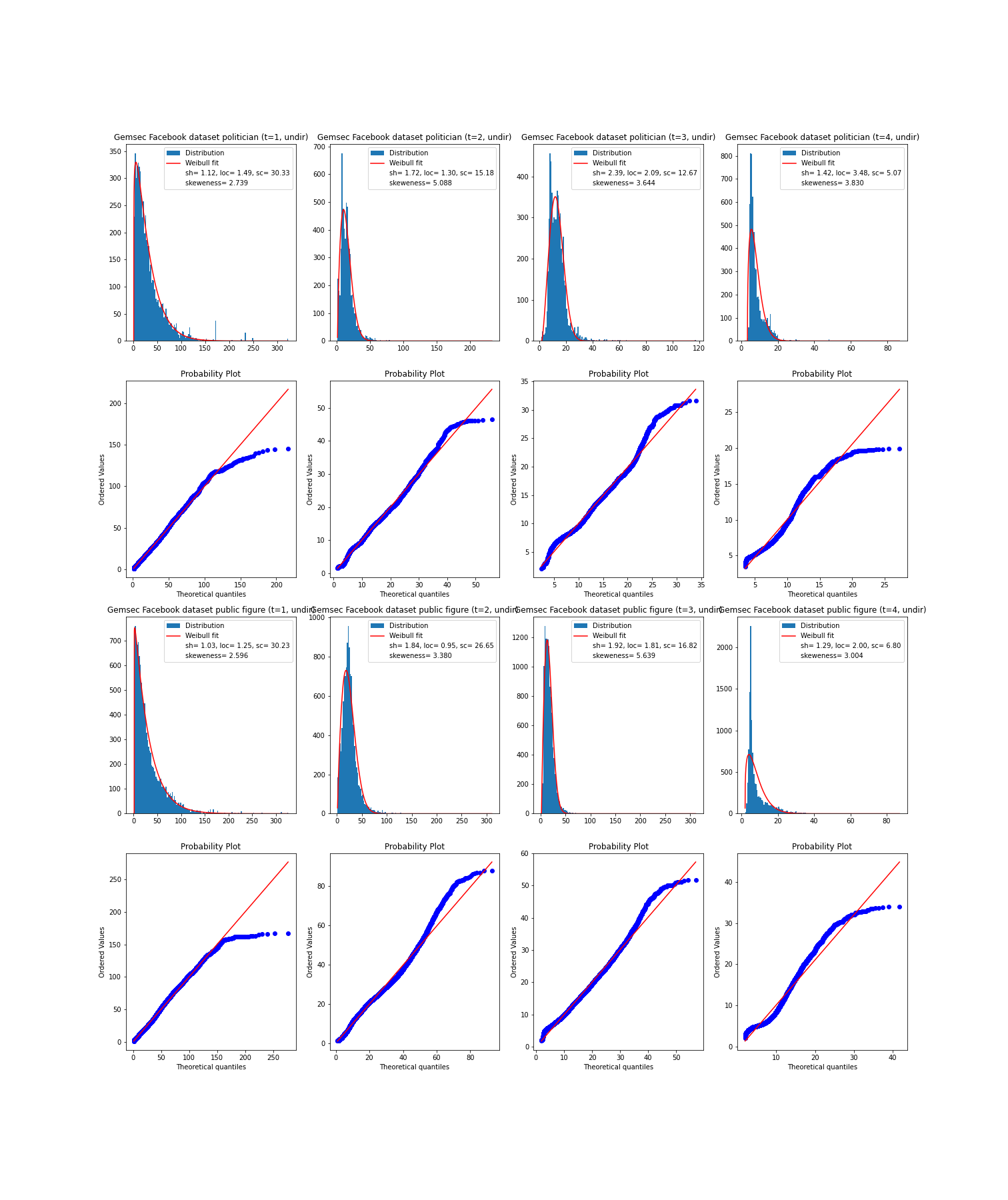}
    \caption{Distributions of $\xi^j$ centralities for $t=j=1,2,3,4$ for real networks data from~\cite{Tuz} fitted with Weibull distribution and their Pearson's moment coefficients of skewness.}
\end{figure}

\begin{figure}[h]
    \centering
    \includegraphics[width=\linewidth]{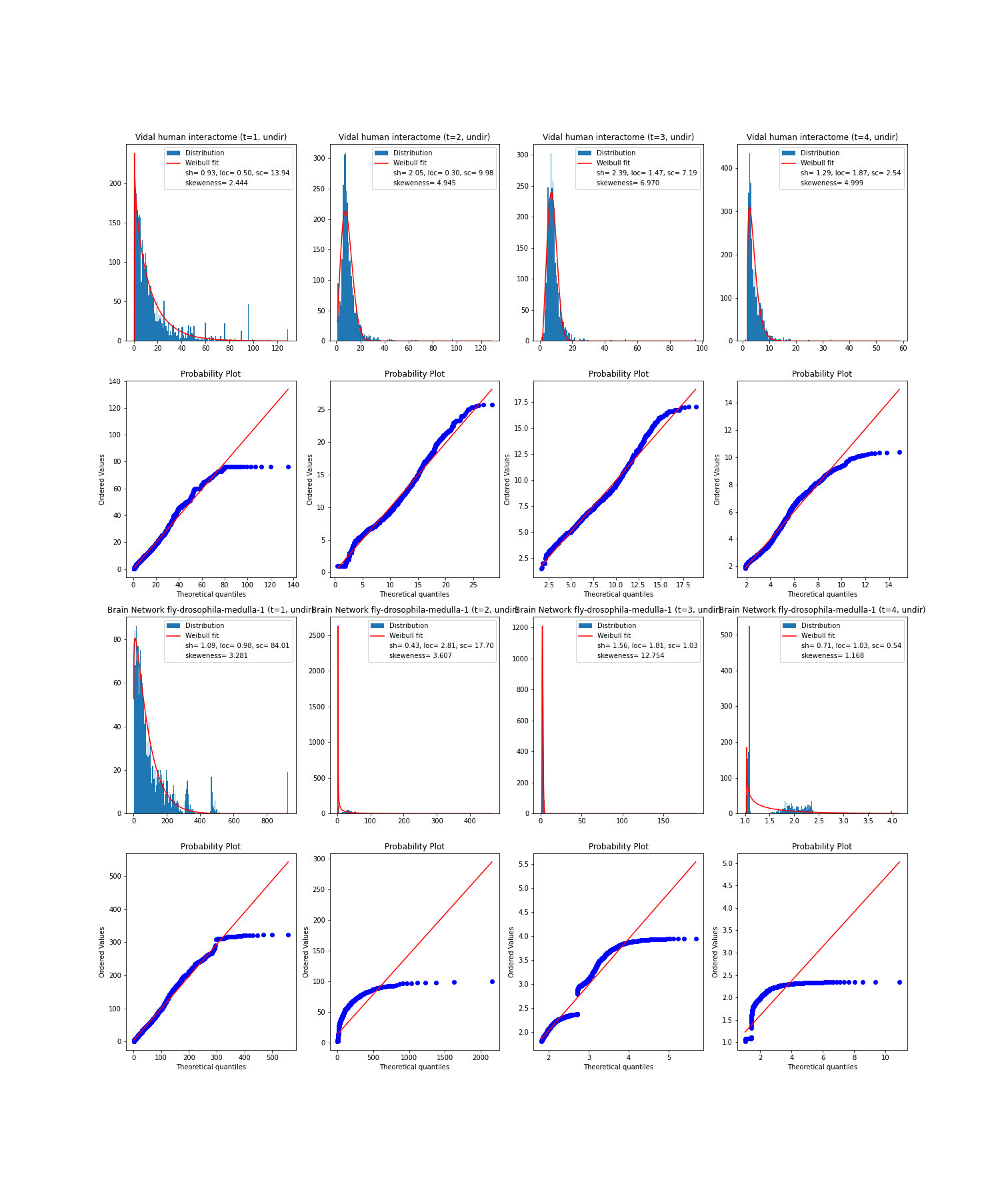}
    \caption{Distributions of $\xi^j$ centralities for $t=j=1,2,3,4$ for real networks data from~\cite{Tuz} fitted with Weibull distribution and their Pearson's moment coefficients of skewness.}
    \label{f_last}
\end{figure}


\begin{thebibliography}{99}

\bibitem{BA1}
Albert, R., Jeong, H., Barabási, A. L. (1999). Diameter of the world-wide web. nature, 401(6749), 130-131.

\bibitem{BA2}
Barabási, A. L., Albert, R. (1999). Emergence of scaling in random networks. science, 286(5439), 509-512.

\bibitem{BH}
Boccaletti Stefano, D-U. Hwang, and Vito Latora. "Growing hierarchical scale-free networks by means of nonhierarchical processes." International Journal of Bifurcation and Chaos 17.07 (2007): 2447-2452.

\bibitem{Tuz} Tuzhilin M. "Capability centrality: the next step from scale-free property." arXiv preprint arXiv:2605.03796 (2026).


\end{thebibliography}
\end{document}